\documentclass{article}

\usepackage[utf8]{inputenc}

\usepackage[top=1in, bottom=0.9in, left=0.9in, right=1in]{geometry}
\usepackage{parskip} % Package to tweak paragraph skipping
\usepackage{amssymb}
\usepackage{makecell}
\usepackage{fullpage}
\usepackage[dvips]{color}
\usepackage{xcolor,float}
\usepackage{latexsym}
\usepackage[activeacute,english]{babel}
\usepackage{graphicx}
\graphicspath{{./figs/}}
\usepackage{amsmath,amsfonts,amssymb,amscd,amsthm}
\usepackage{bm}
\usepackage{nccmath}
\usepackage{algorithm}
\usepackage{natbib}

\usepackage{makecell}

\usepackage{algorithmic}
\usepackage{algorithm}
\usepackage{mathtools}
\usepackage{thmtools}

\usepackage{fancyvrb}

\usepackage{longtable}

\usepackage[style=iso]{datetime2}

\newtheorem{theorem}{Theorem}[section]

\newcommand\R{{\mathbb{R}}}

\newcommand\Z{{\mathcal{Z}}}

\newcommand\F{{\hat F}}

\newcommand\One{{\mathbf{1}}}

\newcommand\epsi{{\varepsilon, i}}

\newcommand\Err{{\rm Err}}
\newcommand\OOB{{\rm OOB}}
\newcommand\Errone{{\widehat{\Err}^{(1)}}}
\newcommand\ErrOOB{{\widehat{\Err}^{\OOB}}}
\newcommand\SE{{\rm SE}}
\newcommand\SEhat{{\widehat{\SE}}}

\newcommand\zalphatwo{{z_{1-\alpha/2}}}

\mathtoolsset{showonlyrefs}

\newcommand{\average}[2][n]{\ensuremath{\dfrac{1}{#1} \sum_{i = 1}^{#1} ~ {#2} }}

\renewcommand\l{\left}
\renewcommand\r{\right}

\newcommand{\tcr}[1]{\textcolor{red}{#1}}

\title{Confidence Intervals for the Generalisation Error of Random Forests}
\author{Rajanala Samyak\thanks{Dept. of Statistics, Stanford Univ.; samyak@stanford.edu},
	 Stephen Bates\thanks{Depts. of Statistics and EECS, Univ. of California, Berkeley; stephenbates@berkeley.edu}, 
	 Trevor Hastie\thanks{Depts. of Statistics and Biomedical Data Science, Stanford Univ.; hastie@stanford.edu}, 
	 and Robert Tibshirani\thanks{Depts. of Biomedical Data Science and Statistics,
    Stanford Univ.; tibs@stanford.edu}}

\date{\today}

\begin{document}
	\maketitle
	
	% ================== Abstract
	\begin{abstract}
		Out-of-bag error is commonly used as an estimate of generalisation error in ensemble-based learning models such as random forests.  
		We present confidence intervals for this quantity using the delta-method-after-bootstrap and the jackknife-after-bootstrap techniques.  
		{\em These methods do not require growing any additional trees.}
		We show that these new confidence intervals have improved coverage properties over the na\"ive confidence interval, in real and simulated examples.
	\end{abstract}
	
	% ================== Introduction
	
	\section {Introduction} \label{sec:intro}
	
	Bootstrap aggregation or bagging is a popular tool for reducing the variance in a learning model by averaging multiple predictions, each of which typically has low bias and high variance. 
	Random Forests \citep{breiman2001random} is a generalization of bagging that uses an ensemble of decision trees, where each tree is grown on a bootstrap sample drawn from the training data,  and only a random subset of the features  are considered at each tree split.
	For free, one also obtains a quantity called the \textit{out-of-bag error}, which provides an estimate of generalisation error.  The out-of-bag error is computed by aggregating the prediction error for observations that were not used in a particular tree.
	Here we extend this idea to obtain a standard error and confidence intervals for the generalisation (test)  error, that is, the error of the error.
	
	In the case of random forest regression, we first describe a na\"ive confidence interval that treats the errors on different observations as independent and examine the coverage properties of this interval. 
	It turns out that this approach tends to undercover in practice, as illustrated in Figure \ref{fig1_inadequacy}  (details are in the Figure caption).  This is a result of the fact that each observation is ``re-used'' -- that is, it plays the role of both a training and a test point.  The same phenomenon occurs for cross-validation as discusssed in \citet{bates2021cross}.
	
	\begin{figure}[H]
		\begin{center}
			\includegraphics[width=3in]{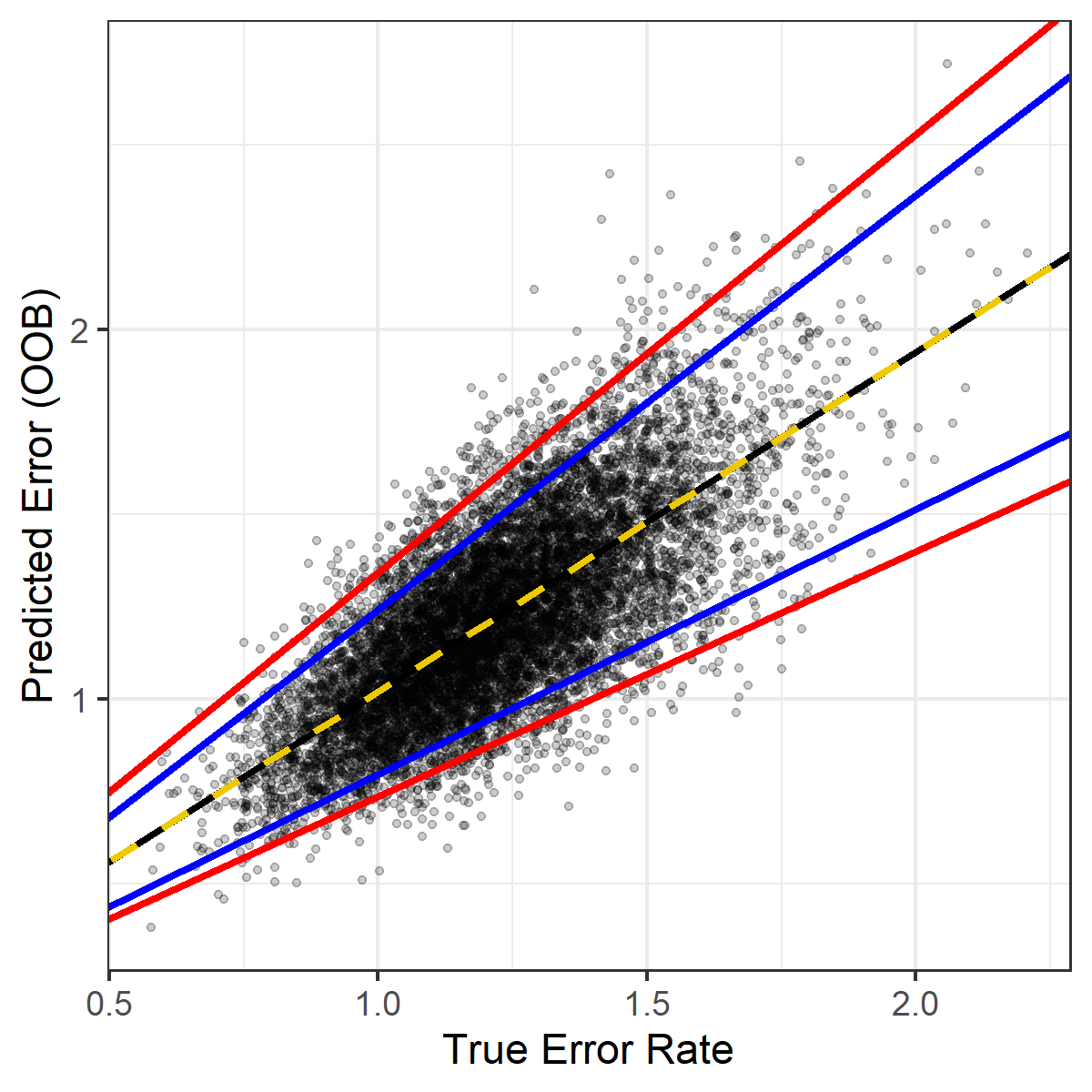}
			\caption{\em An example to illustrate the miscoverage of the na\"ive confidence interval for the true error rate. Each point represents a simulation with $n = 110$ and $p = 1000$.  The striped line is fit to the midpoints and the blue lines are fit to the upper and lower endpoints of the 90\% CIs using the na\"ive SE.  The red lines are obtained using 5\% and 95\% quantile regression.  Ideally, the blue lines would approximate the red lines.
				\label{fig1_inadequacy}}
		\end{center}
	\end{figure}
	
	To remedy this, we study and propose two new methods of computing a confidence interval, without growing any new trees.  We use the \textit{delta-method-after-bootstrap} and the \textit{jackknife-after-bootstrap} of \citet{efron1992jackknife} to obtain estimates of standard error for out-of-bag error, and show that the resultant confidence intervals have better coverage properties than the na\"ive interval.  
	The delta-method-after-bootstrap is also called the \textit{infinitesimal-jackknife-after-bootstrap} \citep{jaeckel1972infinitesimal} and uses influence functions to obtain standard errors for statistics that are smooth functions of the training data.	The jackknife-after-bootstrap uses a novel analogue of the jackknife that exploits the nature of random forests to deliver an estimate of  standard error without growing additional trees.   These two methods have also been used to obtain accuracy measures for other quantities, such as in \citet{efron1997improvements, wager2014confidence}.
		
	In our case of generalisation error in random forests, we obtain expressions for the standard error of out-of-bag error, and use them with a  normal approximation to produce confidence intervals.  
	Our estimator accounts for the fact that the errors on different points are correlated.	
	In Section \ref{sec:review}, we set up notation and review the random forest algorithm.  
	In Section \ref{sec:methods}, we introduce our estimators for the standard error in the regression case, and study the analogue of our result in the case of classification in Section \ref{sec:classification}.   
	In Section \ref{sec:data}, we present results on real and simulated examples.

	\section{Setup and Review} \label{sec:review}
	
	Before turning to our main method in the next section, we introduce our notation and review topics related to out-of-bag error.
	
	\subsection{Notation}
	We consider the standard setup of supervised learning with a real-valued	response.  
	Let $\Z = \{z_1, \dots, z_n\}$ denote the training data where $z_i = (x_i, y_i), i = 1, \dots n$ are drawn i.i.d. from a distribution $P$ on $\R^p \times \R$.
	Let $z_{n+1} = (x_{n+1}, y_{n+1})$ be another independent test point.  Using the training data, we are interested in learning the prediction function $r_\Z : \R^p \mapsto \R$, which minimises the out-of-sample loss $\ell\l(y_{n+1}, r_\Z(x_{n+1})\r)$.  
	
	We will first consider the squared error loss $\ell(a,b) = (a-b)^2$, but our results easily extend to other differentiable loss functions as well.  Note that $\ell\l(y_{n+1}, r_\Z(x_{n+1})\right)$ is random and unknown, so our target is one of two quantities:
	\begin{align}
		\Err_\Z &= E\l[\ell\l(y_{n+1}, r_\Z (x_{n+1})\right) | \Z \right] \\
		\Err &= E\l[\ell\l(y_{n+1}, r_\Z (x_{n+1})\right)\right] = E[\Err_\Z].
	\end{align}
	In practice, we usually consider $\Err_\Z$, which is the test error conditioned on the training data and also called the \textit{true error rate}, to be the estimand of interest.  $\Err$ is also called the \textit{expected true error} and is sometimes used as a quantity of interest.

	\subsection{Bagging \& Out-of-bag Errors} \label{ss:bagging}
	
	The bagging predictor uses an ensemble of $B$ decision trees, in which each decision tree $b = 1, \dots B$ is trained on a bootstrap sample drawn from the training set.  
	The individual trees may have high variance but have low bias.  
	The predictions from these trees are then averaged (\textit{bootstrap aggregation} or \textit{bagging}), resulting in an overall predictor with lower variance.  
	We will use $b$ to index both the decision tree as well as the bootstrap sample on which the decision tree was grown.
	Random Forests \citep{breiman2001random} is a widely used extension of the bagging predictor, where the decision trees are allowed to also depend on an extra source of randomness to encourage more diversity among the trees.  
	The most common implementation for this randomness involves selecting a subset of covariates to be used for creating the splits of the decision tree.
	
	Let $\hat{y}_i^{(b)}$ be the prediction for $x_i$ from tree $b$.  
	As each tree is trained on a sample of size $n$ drawn with replacement from the training set, on average a fraction $(1 - 1/n)^n \approxeq 1/e \approx .368$ of the trees do not use any given observation $i$.
	For each $i$ in the training set, let $o(i)$ be the set of decision trees which did not use $i$.    
	The predictions in $o(i)$ are aggregated to form $\hat{y}^\OOB_i$, the \textit{out-of-bag prediction} for observation $i$, and the \textit{out-of-bag error} is
	\begin{equation}
		\ErrOOB  = \average{\ell(y_i, \hat{y}^\OOB_i)}.
		\label{def:oob}
	\end{equation}
	In the case of regression, $\hat{y}^\OOB_i$ is the average of the predictions in $o(i)$, and in classification, $\hat{y}^\OOB_i$ is determined by majority vote.
		
	$\ErrOOB$ is commonly used as a point estimator of the test error $\Err_\Z$.  In the regression case with square loss, this becomes 
	
	\begin{equation}
		\ErrOOB = \average{ \l(y_i - \hat{y}^{OOB}_i \right)^2 }, \quad \text{where} ~ \hat{y}^{\OOB}_i = \dfrac{1}{|o(i)|}\sum_{b \in o(i)} {\hat{y}^{(b)}_i}.
	\end{equation}
	
	$\ErrOOB$ is instrinsically linked to $n$-fold or leave-one-out cross-validation (LOOCV), and one can show that in the limit as $B \rightarrow \infty$, $\ErrOOB$ is almost equal to the LOOCV error estimate, except for a leading factor of $1/n$ instead of $1/(n-1)$ (See Chap. 15 of \citet{ESL}). 
	A recent article \citep{bates2021cross} provides a careful analysis of cross-validation including understanding the estimand in CV and a nested cross-validation scheme for estimating prediction error.  
	
	So far, we have implicitly assumed that all observations have equal weight.  However, the idea of out-of-bag error can be extended to observations with unequal weights, and indeed, this will be essential for our standard error methods in Section \ref{sec:methods}.  
	
	Let $\F$ be the empirical distribution function of the training data $\Z$.  
	Let $S$ be a \textit{statistic}, by which we mean a real-valued functional that takes as its input a distribution on the training data points.  
	We choose $S$ in a natural way such that $S(\F)$ is the out-of-bag error.

	Let $W = (w_1, \dots, w_n): \sum{w_i} = 1$ be a vector of observation weights for the training data $\Z$, and let $\F_W$ be the corresponding distribution which assigns weight $w_i$ to the $i$-th observation.  Then we have that the empirical distribution $\F = \F_{W_0}$ where $W_0 = (1/n, \dots, 1/n)$.
	
	Let the $B$ trees in the random forest be indexed by $b$. 
	Each tree $b$ corresponds to a bootstrap sample drawn from the training data, so we use $b$ to index both the tree and the bootstrap sample.  
	Let $\widehat{y_j}^{(b)}$ be the prediction from tree $b$, which we will assume is deterministic given the bootstrap sample, % This is something for future work
	and let $I_j^{(b)} = 1$ if observation $j$ is not present in sample $b$, and $0$ if it is present.
	We define the statistic $S$ in the following way:  
	
	\begin{equation}
		S\l(\F_W\r) =
		\sum_{j=1}^{n}{
			w_j
			\l(y_j-\frac{\sum_{b}{
					{\widehat{y_j}}^{(b)}I_j^{(b)}g_W(b)}}
			{\sum_{b}\ I_j^{\left(b\right)}g_W(b)}\right)^2}
	\end{equation}
	where $g_W(b)$ is the probability of drawing the bootstrap sample $b$ under the weights $W$ for the training data.  Under $W_0$, we note that all samples $b$ have the constant likelihood $1/n^n$.
	
	Using this definition, we note that $S(\F) = \ErrOOB$, which is what we want.  This definition also ensures that when some observations have weight zero, we have the expected behaviour in terms of omitting the observation from the entire analysis.

	\subsection{Na\"ive Confidence Interval}	
	As a first attempt to obtain confidence intervals for test accuracy, we could na\"ively suppose that the out-of-bag errors for different observations are independent.  
	Let $q_i = \ell\l(y_i, \hat{y}^\OOB\r) = \l(y_i - \hat{y}^{\OOB}_i \r)^2$, then $\ErrOOB = \bar q = 1/n\sum_i {q_i}$.  Since we have assumed that the $q_i$ are independent, we can use the usual expression for standard error of the mean:
	
	\begin{equation}
		\label{def:naive}
		\SEhat_\text{na\"ive} = \dfrac{1}{\sqrt n}\sqrt{\average{\l(q_i - \bar q \right)^2}}.
	\end{equation}

	We then use the normal approximation to generate confidence intervals.  A $100(1-\alpha)\%$ confidence interval is given by $$\l(\ErrOOB - \zalphatwo \cdot \SEhat_\text{na\"ive}, \ErrOOB + \zalphatwo \cdot \SEhat_\text{na\"ive}\r)$$ where $\zalphatwo$ is the $(1-\alpha/2)^\text{th}$ quantile of the normal distribution.
	
	However Figure \ref{fig1_inadequacy} illustrates that this confidence interval does not have good coverage properties in practice.  In Section \ref{sec:methods}, we instead derive two alternative methods for estimating the standard error of $\ErrOOB$: one using the \textit{delta-method-after-bootstrap} following the example of \citet{efron1997improvements}, and another adapting the jackknife to the specific case of Random Forests, which we term the \textit{jackknife-after-bootstrap}
	\citep{efron1992jackknife}.
	
	\subsection{Related Work}
	
	More broadly, the out-of-bag estimator is one estimator for the general problem of estimating generalization accuracy~\citep[e.g.,][]{efron2021resampling}. There are three main approaches to this problem.  
	First, there are bootstrap-based methods~\citep{efron1983estimating, efron1986, efron1997improvements, efron1993introduction}. 
	Second, there is cross-validation~\citep{Allen74,geisser1975,stone1977} and data splitting. 
	The final main category of prediction error estimates are based on analytic adjustments such as Mallow's $C_p$ \citep{mallows1973comments}, AIC \citep{akaike1974aic}, BIC \citep{schwarz1978bic}, and general \emph{covariance penalties} \citep{stein1981estimating,efron2004estimation}. 
	Our present investigation should be viewed primarily as falling within the first category, but we note that out-of-bag accuracy for random forests is also related to leave-one-out cross-validation (See Chap. 15 of \citet{ESL}).
	
	One such related approach in the first category is the ``leave-one-out-bootstrap'' of \citet{efron1997improvements}.    
	This method is applicable to an arbitrary model fitting procedure, including Random Forests \citep{breiman2001random}.  
	This estimator, denoted $\Errone$, averages the error of models fit to bootstrap samples to derive an estimate of the generalisation error.
	\citet{efron1997improvements} also proposes a standard error for this estimate.
	However, computing $\Errone$ or  its standard error for Random Forests requires fitting a Random Forest to bootstrap samples from the training data, and since each fit itself involves resampling, we obtain a nested bootstrap regime.  
	
	We take inspiration from this approach to derive direct expressions for the standard error of $\ErrOOB$ that does not require a nested bootstrap.   	
	In this paper, we consider methods based on the infinitesimal jackknife and the jackknife for bagging \citep{efron1992jackknife, efron2014estimation}.  These have been studied in the context of model predictions for Random Forests by \citep{wager2014confidence}.  \citet{kim2020predictive} also introduces the \textit{jackknife+-after-bootstrap} for predictive intervals.  \citet{giordano2020swiss} presents theoretical results and error guarantees for the infinitesimal jackknife in general situations.  \citet{athey2019generalized} contains a literature review of other techniques related to Random Forests.

	\section{Methods} \label{sec:methods}

	We now turn to our proposed estimates of standard error.  In Section \ref{sec:delta}, we build on the delta-method-after-bootstrap and propose an estimator $\SEhat_{\rm del}$ for the standard error of $\ErrOOB$.  In Section \ref{sec:jack}, we adapt the jackknife estimator of standard error to the case of Random Forests and propose a jacknife-after-bootstrap estimator $\SEhat_{\rm JAB}$.

	\subsection{Delta-method-after-bootstrap} \label{sec:delta}
	
The delta-method-after-bootstrap, also known as the infinitesimal jackknife \citep{jaeckel1972infinitesimal, efron1992jackknife} can be used to derive estimates of accuracy for statistics which are ``smooth'' functions of $\F$.  

We will show that $\ErrOOB$ is also a ``smooth'' function of $\F$ and derive an expression for the standard error of $\ErrOOB$, following the outline of \citet{efron1995cross, efron1997improvements}.

What do we mean by a ``smooth'' function of $\F$?  
$\F$ is a distribution on the training data which puts an equal weight $1/n$ on each of the $n$ training data points.  
Let $\F_\epsi$ be the distribution obtained by perturbing the weight of observation $i$ by $\varepsilon$, i.e.

\begin{equation}
	\hat{F}_\epsi : {\rm Pr}
	\begin{cases}
		\dfrac{1-\varepsilon}{n} + \varepsilon &\text{on}~z_i \\
		\\
		\dfrac{1-\varepsilon}{n} &\text{on}~z_j~\text{for}~j\neq i.
	\end{cases}
\end{equation}

Then we say a symmetrically defined statistic $S(\F)$ is ``smooth'' if the derivatives $\partial S(\F_\epsi)/\partial\varepsilon$ exist at $\varepsilon = 0$.

Defining
\begin{equation}
	\hat{D}_i = \dfrac{1}{n} \dfrac{\partial S(\F_\epsi)}{\partial \varepsilon} |_0,
	\label{def:Di}
\end{equation}

the nonparametric delta method standard error estimate for $S(\hat{F})$ is 
\begin{equation}
	\label{def:delta}
	\widehat{\rm SE}_\text{del} (S) = \l[\sum_1^n{\hat{D}_i^2} \r]^{1/2}
\end{equation}
(see \cite{efron1992jackknife}, Section 5).  The vector $\hat{\bf D} = (\hat{D}_1, \dots, \hat{D}_n)$ is $1/n$ times the \textit{empirical influence function} of $S$.

We now present the main result for the delta-method-after-bootstrap:
\begin{theorem} \label{lem:Di}
	Let $B$ be the total number of distinct trees, which is $n^n$ in the case of bagging.  Let $n^{(b)}_i$ be the number of times observation $i$ occurs in sample $b$, and $I^{(b)}_j = \One{\l(n^{(b)}_j = 0\r)}$.  For $S(\F) = \ErrOOB$ with the square loss, the derivative (\ref{def:Di}) is 
	\begin{multline}
		{\hat{D}}_i=
		\frac{1}{n}
		\l\{
		\l( y_i - {\widehat{y_i}}^{\OOB} \r)^2 -
		\dfrac{1}{n}\sum_j\l(y_j-{\widehat{y_j}}^{\OOB}\right)^2
		\r\} - \\
		\frac{2e_n}{n}
		\sum_{j}{
			\l(y_j - {\widehat{y_j}}^{\OOB}\r) \cdot
			\l\{
			\frac{1}{B} \sum_{b} {
				\l(N_i^{\l(b\r)}-1\r) 
				I_j^{\l(b\r)} 
				\l( {\widehat{y_j}}^{\l(b\r)} - 
				{\widehat{y_j}}^{\OOB} \r)
			}
			\r\}},
	\end{multline}
	where $e_n = (1-1/n)^{-n}$.
\end{theorem}
We defer the proof to the appendix.  % The proof is an application of the chain rule for differentiation.

We can now evaluate $\SEhat_\text{del} = \sqrt{\sum_{i}\ D_i^2}$.  We compare this to the na\"ive SE of (\ref{def:naive}).  We see that 
\begin{equation}
	\SEhat_\text{del}=\frac{1}{n}\sqrt{{\sum_{i}\l(q_i-\bar{q}+C_i\r)}^2}.
\end{equation}
where $$ C_i=
-\frac{2e_n}{n}
\frac{1}{B}\sum_{b}
\l(N_i^{\l(b\r)}-1\r)
\l\{
\sum_{j}{
	I_j^{\l(b\r)} 
	\l(y_j - {\widehat{y_j}}^{\OOB}\r) 
	\l({\widehat{y_j}}^{\l(b\r)} - {\widehat{y_j}}^{\OOB}\r)
}
\r\}$$ is proportional to the bootstrap covariance between $N_i^{\l(b\r)}$ and the cross term
$\sum_{j}{
	I_j^{\l(b\r)}
	\l(y_j-{\widehat{y_j}}^{\OOB}\r)
	\l({\widehat{y_j}}^{\l(b\r)} - {\widehat{y_j}}^{\OOB}\r)
}$.  
The na\"ive SE results from taking $C_i = 0$.

In practice, the number of trees $B$ is far less than the total number of distinct trees, which is $n^n$ in the case of bagging and could depend on the exact sampling scheme of the Random Forest.  Hence, following \citet{efron1997improvements}, we replace the expected value in $(N_i^{(b)} - 1)$ by the sample average  $\l(N_i^{\l(b\r)}-{\bar{n}}_i\r)$ in the covariance expression to get
\begin{multline}
	{\hat{D}}_i = 
	\frac{1}{n}
	\l\{
	\l(y_i-{\widehat{y_i}}^{\OOB}\r)^2 -
	\dfrac{1}{n}\sum_j \l(y_j-{\widehat{y_j}}^{\OOB}\r)^2
	\r\}- \\
	\frac{2e_n}{n}\sum_{j}{
		\l(y_j-{\widehat{y_j}}^{\OOB}\r) \cdot 
		\l\{\frac{1}{B}\sum_{b}{
			\l(N_i^{\l(b\r)}-{\bar{n}}_i\r)
			I_j^{\l(b\r)}
			\l({\widehat{y_j}}^{\l(b\r)}-{\widehat{y_j}}^{\OOB}\r)
		}
		\r\}}.
\end{multline}

We note that typically  $\SEhat_\text{del} \geq \SEhat_\text{na\"ive}$ but it is possible  to have $\SEhat_\text{del} < \SEhat_\text{na\"ive}$ for certain regimes of $q_i$ and  $C_i$.  

For a conservative estimate, we can define
\begin{equation}
	\label{def:pmax}
	\SEhat_\text{del+}=\max\l\{\SEhat_\text{na\"ive}, \SEhat_\text{del}\r\}.
\end{equation}

The $100(1-\alpha)\%$ confidence interval is given by $$\l(\ErrOOB - \zalphatwo \cdot \SEhat_\text{del+}, \ErrOOB + \zalphatwo \cdot \SEhat_\text{del+}\r)$$ where $\zalphatwo$ is the $\alpha/2^\text{th}$ quantile of the normal distribution.

For an arbitrary differentiable loss function $\ell(a,b)$, we have the analogue of Theorem \ref{lem:Di}:

\begin{theorem} \label{lem:arb_loss}
	Using the notation of Theorem \ref{lem:Di}, For $S(\F) = \ErrOOB$, with arbitrary loss function $\ell$ which is continuous and differentiable almost everywhere, the derivative (\ref{def:Di}) is 
	\begin{multline}
		{\hat{D}}_i=
		\frac{1}{n}
		\l\{
		\ell\l( y_i , {\widehat{y_i}}^{\OOB} \r) -
		\dfrac{1}{n}\sum_j\ell\l(y_j, {\widehat{y_j}}^{\OOB}\right)
		\r\} + \\
		\frac{e_n}{n}
		\sum_{j}{
			\ell'\l(y_j, {\widehat{y_j}}^{\OOB}\r) \cdot
			\l\{
			\frac{1}{B} \sum_{b} {
				\l(N_i^{\l(b\r)}-1\r) 
				I_j^{\l(b\r)} 
				\l( {\widehat{y_j}}^{\l(b\r)} - 
				{\widehat{y_j}}^{\OOB} \r)
			}
			\r\}},
	\end{multline}
	where 
	$\ell'\l(a,b\r) = 
	\dfrac{\partial 
		\ell\l(a, b\r)}{
		\partial b}$.
\end{theorem}

Some common cases of $\ell$ and $\ell'$ are listed in Table \ref{tab:ell}.

\begin{table}[H]
	\begin{center}
		\begin{tabular}	
			{||c|c|c||} 
			\hline
			Loss & $\ell\l(a,b\r)$ & $\ell'\l(a,b\r) = \partial \ell\l(a, b\r)/\partial b$ \\ 
			\hline\hline
			Square error ($l^2$) & $(a-b)^2$ & $-2(a-b)$ \\ 
			\hline
			Absolute error ($l^1$) & $|a-b|$ & $-\text{sign}(a-b)$ \\
			\hline
			\makecell{Binomial deviance  \\ ($a, b \in [0,1]$)} & \makecell{$- a\log(b) - (1-a)\log(1-b)$} & $- a/b + (1-a)/(1-b)$  \\ [1ex] 
			\hline
			
		\end{tabular}
	\end{center}
	\caption{\em Some common loss functions and relevant derivatives in the context of Theorem \ref{lem:arb_loss}.
		\label{tab:ell}}
\end{table}

	\subsection{Jackknife-after-Bootstrap} \label{sec:jack}

	The usual jackknife estimate of standard error of a statistic $T$  from $n$ observations is based on the  jackknife quantities
	\begin{equation*}
		T_{(1)}, \dots, T_{(n)}
	\end{equation*}
	where $T_{(i)}$ is the statistic $T$ computed with observation $i$ omitted.  The jackknife estimator for standard error is then given by $$\sqrt{\dfrac{n-1}{n} \sum_{i=1}^n {(T_{(i)} - T_{(\cdot)})^2}}$$ 
	where $T_{(\cdot)} = 1/n\sum_i{T_{(i)}}$.
	
	In the case of Random Forests where we have a fixed number of trees $B$, this means re-fitting an entire random forest with $B$ trees for each left-out observation.  While valid, this is computationally expensive.  For clarity, we call this the \textit{full jackknife}. 
	We instead define the \textit{jackknife-after-bootstrap} standard error $\SEhat_{\text{JAB}}$ which replaces $T_{(1)}, \dots, T_{(n)}$ by 	
	\begin{equation*}
		S(\F_{-1/(n-1), 1}), S(\F_{-1/(n-1), 2}), \dots, S(\F_{-1/(n-1), n}).
	\end{equation*}
	
	The estimator is thus given by
	\begin{equation*}
		\SEhat_{\rm JAB} = \sqrt{\dfrac{n-1}{n} \sum_{i=1}^n {(S(\F_{-1/(n-1), i}) - S(\F_{-1/(n-1), \cdot}))^2}},
	\end{equation*}
	where $S(\F_{-1/(n-1), \cdot}) = 1/n \sum_i S(\F_{-1/(n-1), i})$.

	We observe that $\F_{-1/(n-1), i}$ omits observation $i$ and places an equal weight $1/(n-1)$ on each of the other observations.  Note that the jackknife-after-bootstrap agrees with the full jackknife when all possible trees are in the Random Forest.  However, for fixed $B$, the jackknife-after-bootstrap re-uses the $B$ trees in the original fit whereas the full jackknife regrows $B$ trees $n$ times.  This means we can compute it efficiently with Random Forests.
	
	The $100(1-\alpha)\%$ confidence interval is given by 
	$$\l(\ErrOOB - \zalphatwo \cdot \SEhat_{\rm JAB}, \ErrOOB + \zalphatwo \cdot \SEhat_{\rm JAB}\r)$$
	 where $\zalphatwo$ is the $\alpha/2^\text{th}$ quantile of the normal distribution.
	
	We note that for each left-out observation in the Random Forest, there are on average only $.632 B$ trees in the Random Forest built on the dataset without the observation.  In particular, this suggests that if $B$ is chosen that the out-of-bag error is stabilised for the original forest, then we would need at least $B/.632$ trees if we would like to use $\SEhat_{\rm del}$ or $\SEhat_{\rm JAB}$.
	
	Both $\SEhat_{\rm del}$ and $\SEhat_{\rm JAB}$ can be computed directly from the output of popular packages for fitting Random Forests, such as \texttt{randomForest} and \texttt{ranger} in R, without any modifications to the underlying code.
	
	\subsection{Transformation of intervals} \label{ss:transform}
	
	In Section \ref{sec:data}, we see that often the confidence intervals are asymmetric, in the sense that the miscoverage on either side is not equal. This is partially explained by the fact that the square error causes a skew in the distribution of $\ErrOOB$.  To remedy this, instead of using normal confidence intervals on the original scale, we can consider transformed intervals.  Let $h:\R \rightarrow \R$ be a monotonically increasing function.  Then we can consider the intervals for the quantity $h\l(\ErrOOB\r)$:
	\begin{equation*}
		\l[h(\ErrOOB) - \zalphatwo\cdot h'(\ErrOOB) \cdot \SEhat, \quad
		   h(\ErrOOB) + \zalphatwo\cdot h'(\ErrOOB) \cdot \SEhat \r]
	\end{equation*} which can then be transformed back to the original scale.  We list the most common transformations and the corresponding intervals below.
	
	\begin{align*}
		{\rm Original}, h(x) = x&: 
		\l[\ErrOOB - \zalphatwo\cdot\SEhat, 
		 \ErrOOB + \zalphatwo\cdot\SEhat\r] \\
		{\rm Log}, h(x) = \log(x) &: 
		\l[\exp\l(\log\l(\ErrOOB\r) - \zalphatwo\dfrac{\SEhat}{\ErrOOB}\r), 
		 \exp\l(\log\l(\ErrOOB\r)+ \zalphatwo\dfrac{\SEhat}{\ErrOOB}\r)\r] \\
		 {\rm Square~Root}, h(x) = \sqrt{x} &: 
		 \l[\l(\sqrt{\ErrOOB} - \zalphatwo\dfrac{\SEhat}{2\sqrt{\ErrOOB}}\r)^2, 
		 \l(\sqrt{\ErrOOB} + \zalphatwo\dfrac{\SEhat}{2\sqrt{\ErrOOB}}\r)^2\r] \\
	\end{align*}

	We find the log transformation useful in creating intervals with more symmetric coverage than the original scale, and we present extended results in Appendix \ref{app:results-ext}.

	\section{Random Forest Classification} \label{sec:classification}

We now study the case of two-class classification.  As in the case of regression, we have the training set $\Z = \{z_1, \dots, z_n\}$ where $z_i = (x_i, y_i), i = 1, \dots n$, but now drawn i.i.d. from a distribution $P$ on $\R^p \times \{0, 1\}$.
Again, let $z_{n+1} = (x_{n+1}, y_{n+1})$ be another independent test point.  

In classification, we usually consider the misclassification loss $l(a,b) = \One{(a \neq b)}$. The out-of-bag error (\ref{def:oob}) becomes

\begin{equation}
	\ErrOOB = \average{ \One{(y_i \neq \hat{y}^{\OOB}_i )} },
\end{equation}

where $\hat{y}^{\OOB}_i$ is determined by majority vote from the predictions of the out-of-bag trees.  Treating each prediction $\hat{y}^{(b)}_i$ as a $0/1$ numeric indicator, we can rephrase this as $$\hat{y}^{\OOB}_i = \mathbf{1}\l(\dfrac{1}{|o(i)|}\sum_{b \in o(i)} {\hat{y}^{(b)}_i} > 1/2\r).$$
In other words, $\hat{y}^{\OOB}_i = 1$ if more than half of the out-of-bag trees predict class $1$. 

While $\ErrOOB$ is not necessarily a smooth function of $\hat{F}$ because of the discrete nature of the misclassification loss, we note that the classification error can be rewritten as the error under the squared loss, treating the binary $\{0,1\}$ response as continuous.  We can do this by exploiting the following fact about the misclassification and the square loss:

\begin{equation}
	\ell(a,b) = \One(a \neq b) = (a-b)^2 \qquad {\rm if} ~ a, b \in \{0,1\}.
\end{equation}

This leads us to the following relation:
\begin{equation}
	\ErrOOB= \average{ \One{(y_i \neq \hat{y}^{\OOB}_i )} } = \average{ (y_i - \hat{y}^{\OOB}_i )^2 }.
\end{equation}

This suggests the use of the standard error and the CI from regression in the classification case as well.  This is akin to phrasing the classification case as a $0/1$ regression problem, with the additional thresholding of $\hat{y}^{\OOB}_i$ to make it $0/1$ valued.  To be explicit, we use the following quantity as the delta method standard error estimate:
\begin{equation}
	\label{def:delta_class}
	\widehat{\rm SE}_\text{del} (S) = \l[\sum_1^n{\hat{D}_i^2} \r]^{1/2}
\end{equation}
where using the notation of Theorem \ref{lem:Di},
\begin{multline}
	{\hat{D}}_i=
	\frac{1}{n}
	\l\{
	\l( y_i - {\widehat{y_i}}^{\OOB} \r)^2 -
	\dfrac{1}{n}\sum_j\l(y_j-{\widehat{y_j}}^{\OOB}\right)^2
	\r\} - \\
	\frac{2e_n}{n}
	\sum_{j}{
		\l(y_j - {\widehat{y_j}}^{\OOB}\r) \cdot
		\l\{
		\frac{1}{B} \sum_{b} {
			\l(N_i^{\l(b\r)}-1\r) 
			I_j^{\l(b\r)} 
			\l( {\widehat{y_j}}^{\l(b\r)} - 
			\dfrac{1}{|o(i)|} \sum_{b' \in o(i)}{\hat{y}^{(b')}_i} \r)
		}
		\r\}}.
\end{multline}

	\section{Results} \label{sec:data}
	
	Here we present the results from simulation experiments as well as analysis of real data examples.  We are primarily interested in the coverage of confidence intervals for generalisation error using the two methods we proposed and compare them to the coverage of the na\"ive interval.
	The code to reproduce all the results is available on GitHub at \texttt{https://github.com/RSamyak/oobdelta\_results/}.
	
	\subsection{Simulated Examples: Regression} \label{ss:sim_reg}
	
	We will now explore the behaviour and properties of our methods in various simulation settings.  We first show an example visualisation of the confidence intervals in Figure \ref{fig:cis_visualisation}.  
	The black dots are the out-of-bag error rate $\ErrOOB$, and the red dots are the true error rate in each run of the simulation, which is evaluated on a very large test set.  Confidence intervals centred around $\ErrOOB$ are shown using black lines.  We highlight the intervals which do not cover the true error rate, which we call miscoverage.  
	We define the interval to have miscoverage of the high type if both endpoints of the interval are larger than the true value, and to have miscoverage of the low type if both endpoints are smaller than the true value.  We say we have coverage of the high type when we do not have miscoverage of the high type, and similarly for the low type.

	  \begin{figure}[H]
	  	\centering
	  	\includegraphics[height=3in]{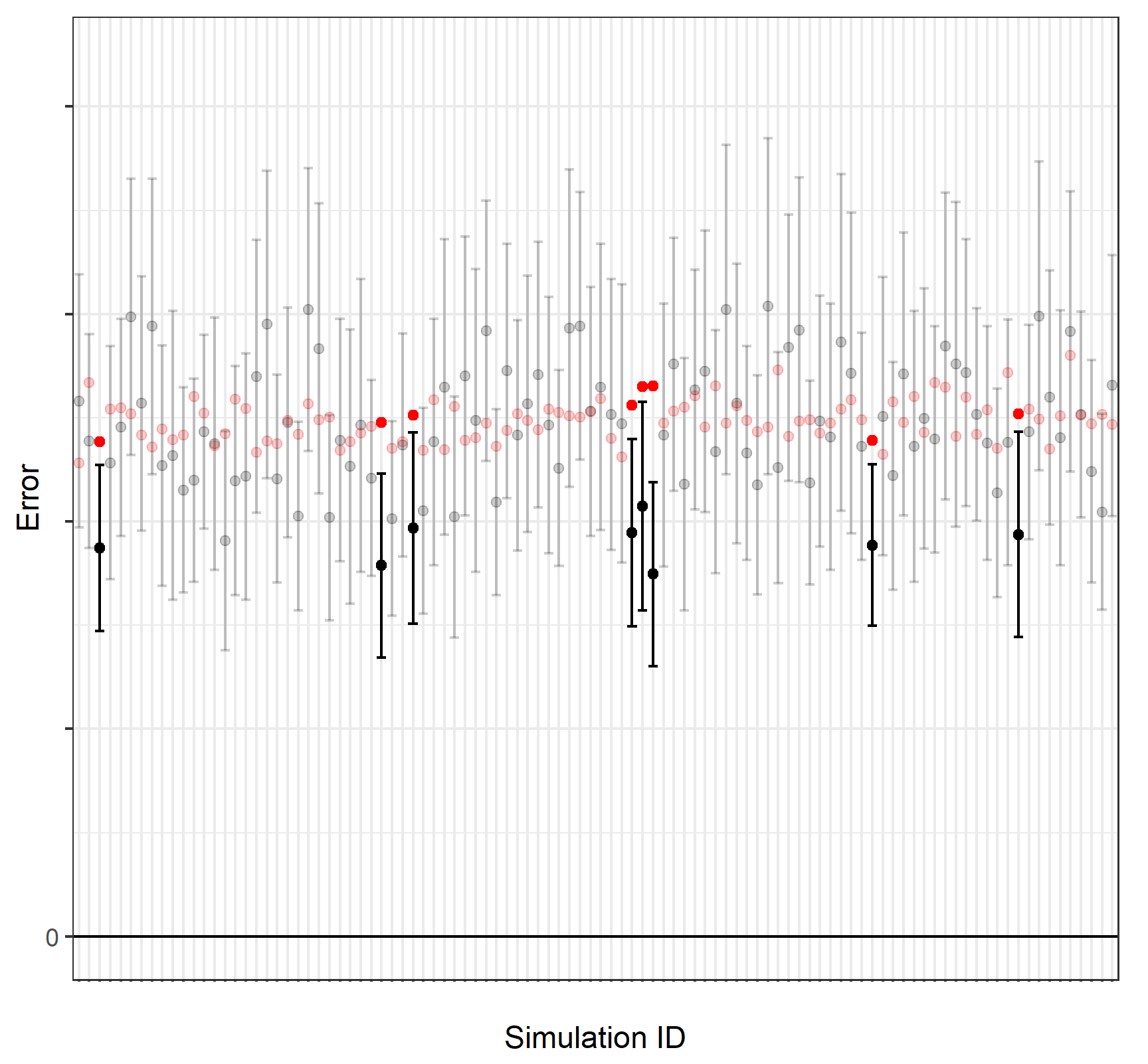}
	  	\caption{\em An example visualisation of the confidence intervals. In this illustration, we look at a sample of 90\% confidence intervals on the original scale using $\SEhat_{\rm JAB}$, with $n = 110, p = 1000, B = 3000, SNR = 2.0$.  The black dots are the out-of-bag error rate, the red dots are the true error rate.  Confidence intervals that do not cover the true error rate are highlighted in black.
	  		\label{fig:cis_visualisation}	}	
	  \end{figure}
	  
	In Table \ref{tab:results_reg}, we report the coverage of the na\"ive interval and the new CIs using the delta-method-after-bootstrap SE (\ref{def:pmax}) and the jackknife-after-bootstrap SE.  We fix the number of trees ($B = 3000$) across the different simulation settings and consider the nominal 90\% intervals.  We also report the average true error rate and the average OOB error estimate.  We observe that the delta method and the JAB intervals are quite comparable, and both have much better coverage properties than the na\"ive interval.  We see that these results are consistent across a range of different scales for the true error rate.

	\begin{table}[H]
		\begin{center}
			\begin{tabular}{|| *{11}{c}|| }
				\hline
				\multicolumn{3}{||c}{Setting}    
				& \multicolumn{2}{c}{Mean}
				& \multicolumn{3}{c}{Mean CI Width}
				& \multicolumn{3}{c||}{Miscoverage}             \\
				$n$ & $p$ & SNR 
				& $\ErrOOB$ & Truth
				& Na\"ive & Delta & JAB
				& {Na\"ive}
				& {Delta}
				& {JAB}                \\ [0.5ex]
			% latex table generated in R 4.1.1 by xtable 1.8-4 package
			% Fri Jan 14 22:40:10 2022
			\hline
			\hline
			110 & 10 & 0 & 1.1 & 1.1 & .47 & .54 & .62 & \tcr{16.4\%} & \tcr{11.6\%} & 7.7\% \\ 
			\hline
			110 & 100 & 0 & 1.0 & 1.0 & .45 & .52 & .55 & \tcr{13.5\%} & 9.0\% & 6.7\% \\ 
			\hline
			110 & 1000 & 0 & 1.0 & 1.0 & .44 & .51 & .53 & \tcr{13.8\%} & 9.3\% & 7.7\% \\ 
			\hline
			110 & 10 & 2 & 76 & 75 & 33 & 38 & 39 & \tcr{13.0\%} & 8.9\% & 7.6\% \\ 
			\hline
			110 & 100 & 2 & 1147 & 1142 & 503 & 580 & 591 & \tcr{11.6\%} & 7.4\% & 6.6\% \\ 
			\hline
			110 & 1000 & 2 & 12421 & 12450 & 5441 & 6255 & 6417 & \tcr{14.4\%} & 8.8\% & 7.9\% \\ 
			\hline
			110 & 10 & 10 & 42 & 42 & 18 & 21 & 20 & \tcr{12.3\%} & 8.6\% & 9.9\% \\ 
			\hline
			110 & 100 & 10 & 804 & 803 & 352 & 406 & 406 & \tcr{12.9\%} & 8.8\% & 8.6\% \\ 
			\hline
			110 & 1000 & 10 & 9062 & 9079 & 3965 & 4556 & 4660 & \tcr{14.0\%} & 9.4\% & 8.2\% \\ 
			\hline
			\hline
			\end{tabular}

		\end{center}
	
	\caption{\em Simulation results for Regression.  In each setting, the coverage is computed over $R = 1000$ replicates.  The Random Forest algorithm is fixed to use $B = 3000$ trees.  CIs are generated with 10\% nominal miscoverage. \label{tab:results_reg}}
		
	\end{table}

	In Figure \ref{fig:reg_cov}, we visualise the coverage of the $100(1-\alpha)\%$ intervals on each side, with varying $\alpha$ across different $p$ for fixed $n$ and $B$.  We observe that the delta method and the JAB intervals consistently have better coverage than the na\"ive interval.  We also observe that the miscoverage of the high type is consistently less than the miscoverage of the low type.

	\begin{figure}[H]
		\centering
		\includegraphics[height=2.6in]{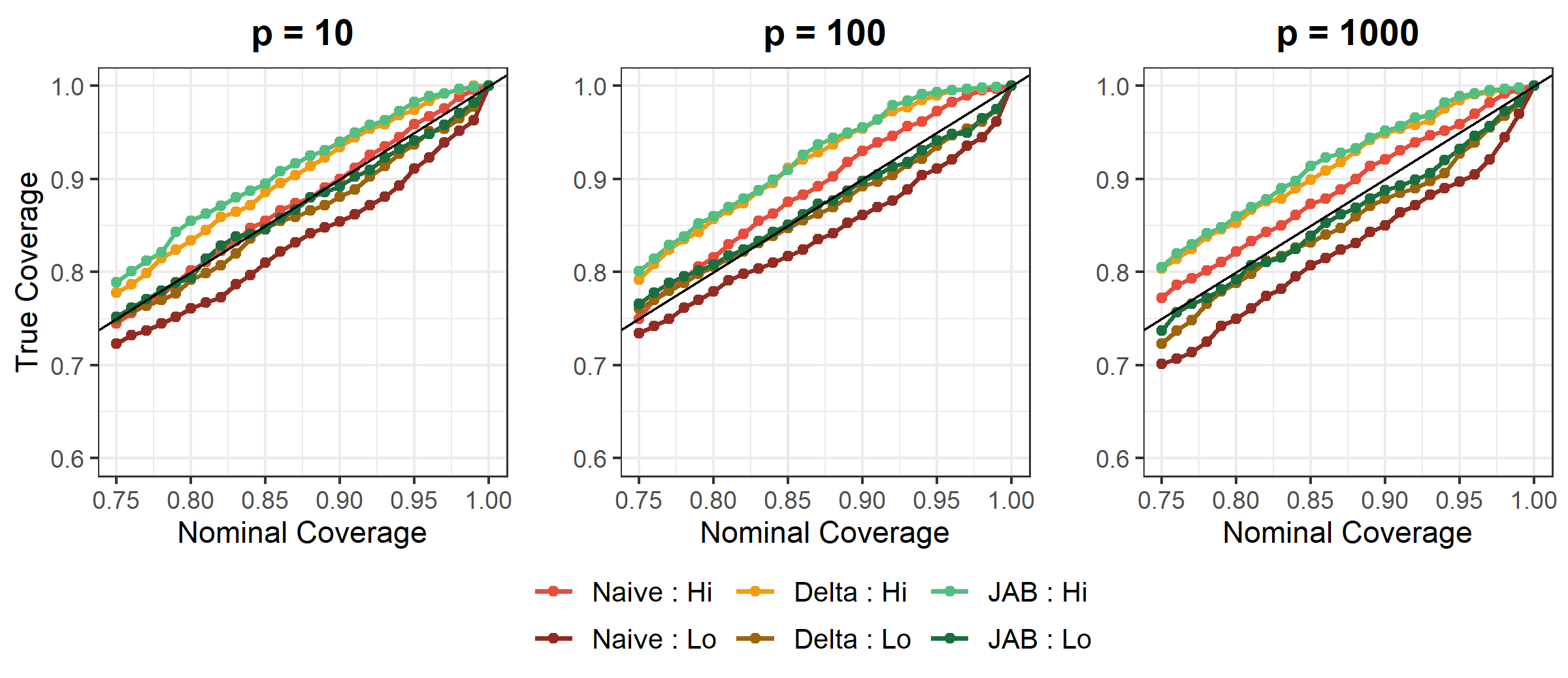}
		\caption{\em Coverage plots under different settings for regression. In each setting, the coverage is computed over $R = 1000$ replicates.  The Random Forest algorithm is fixed to use $B = 3000$ trees.  In this example $n = 110$ and $p$ varies, with the signal-to-noise ratio fixed to be $2.0$.  The true error rate is estimated using a test set of size $n_{\rm test} = 11k$.   CIs are generated with 5\% nominal miscoverage in each tail.
		\label{fig:reg_cov}	}
		
	\end{figure}

	In Figure \ref{fig:reg_lines}, we look at trends in coverage across varying parameters $p$, $B$, and SNR.  We observe that coverage is relatively stable across different $p$ and $SNR$, but we observe a downward trend as $B$ increases, resulting in slight miscoverage as $B$ takes very large values.

	\begin{figure}[H]
		\centering
		\includegraphics[height=2.6in]{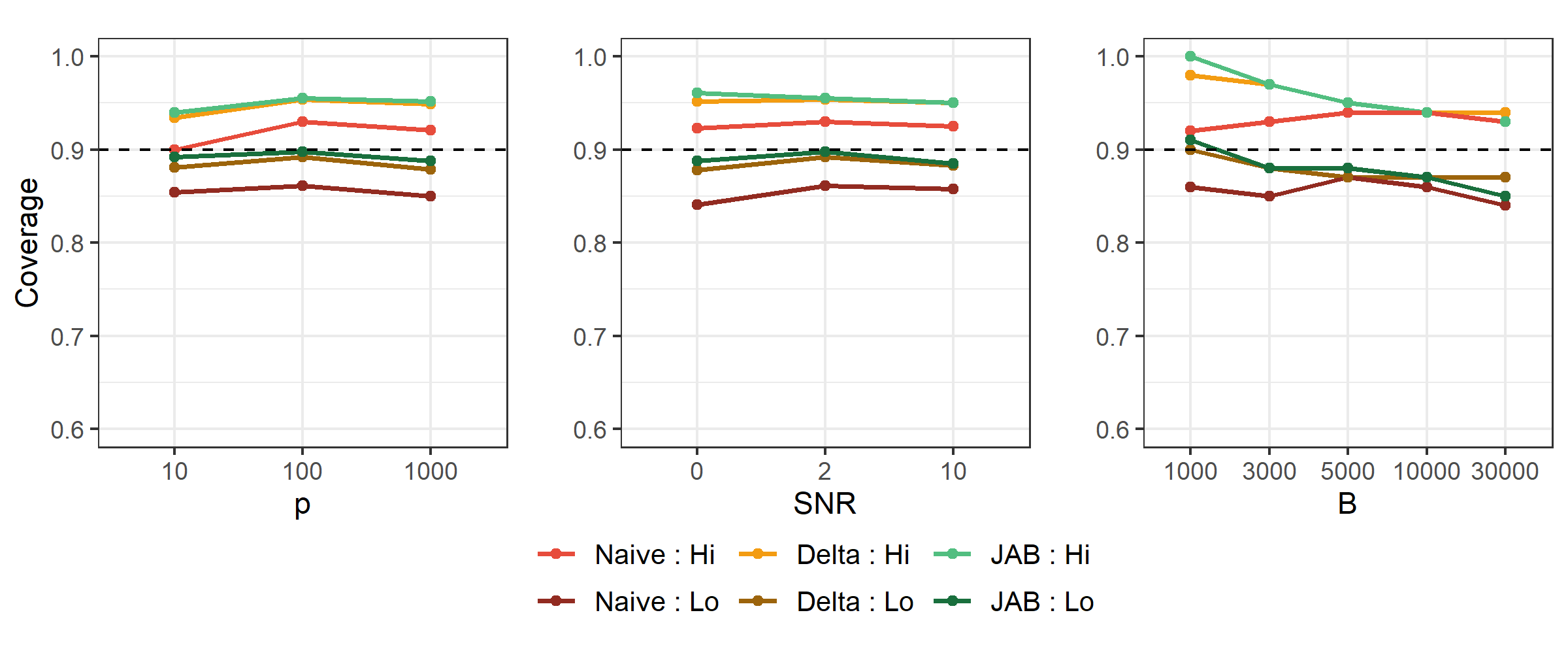}
		\caption{\em Trends across varying parameters $p$, $B$, and SNR for Regression. In each setting, $n = 110$ and the coverage is computed over $R = 1000$ replicates.  
			The true error rate is estimated using a test set of size $n_{\rm test} = 11k$.  Left to right: (1) $B = 3000$, SNR $ = 2.0$, and varying $p$; (2) $p = 100$, $B = 3000$, and varying SNR; (3) $p = 100$, SNR $ = 2.0$, and varying $B$.
			\label{fig:reg_lines}	}	
	\end{figure}
	
	In Figure \ref{fig:reg_sd_curve_B}, we explore the trend in the standard error estimates as $B$ increases in order to understand the trend in coverage observed in Figure \ref{fig:reg_lines}.  
	We notice that the JAB and the delta method estimates for standard error take longer to stabilise, than usual for the point predictions in Random Forests.  
	For small $B$, we observe that the standard error estimates are very large which leads to wide intervals and hence high coverage.
	This is an important observation as $B$ is a parameter in Random Forests that needs to be chosen.  We do not yet fully understand how to pick the right $B$ from first principles in order to obtain reliable standard error estimates.  However, in Section \ref{sec:jack} we show this must be larger than the $B$ required for stabilising the out-of-bag error.  For our experiments we picked $B = 3000$, as we found this to be a reasonable choice.
	
	\begin{figure}[H]
		\centering
		\includegraphics[width=6in]{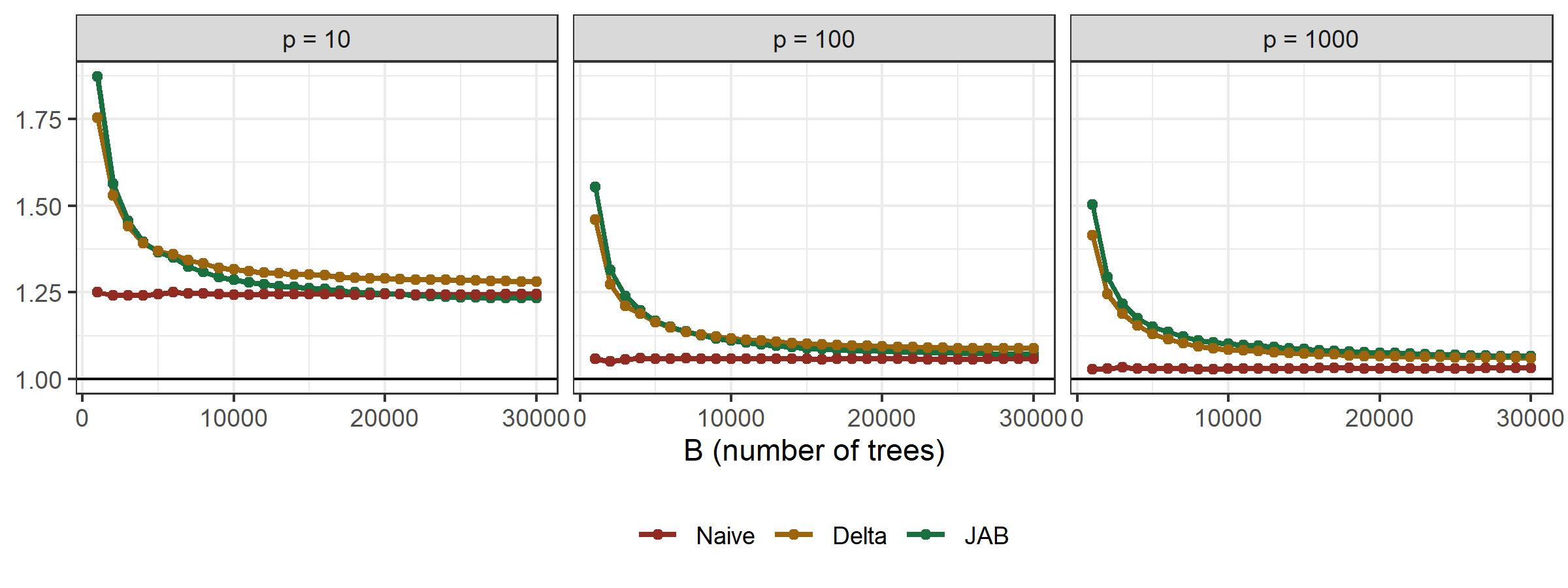}
		\caption{\em Trend in standard error estimates across varying $B$ for Regression, rescaled to show the ratio of the estimates to the empirical SD of $\ErrOOB$.  In each setting, $n = 110, SNR = 2.0$ and the curves are averaged over $R = 100$ replicates.
			\label{fig:reg_sd_curve_B}	}	
	\end{figure}
	
	\subsection{Real Data Examples: Regression} \label{ss:realdata}
	
	We use our method on various datasets obtained from the UCI Machine Learning Repository and show the results obtained in Table \ref{tab:realdata_reg}.  In each case, we randomly split the data into a training set (20\%) and a test set (80\%), and report the average standard error estimates.  The larger test set is essential for accurate estimation of the true error rate.  In Figure \ref{fig:realdata_reg_coverage}, we show coverage under repeated train-test splits of the Communities dataset.  We observe that both the delta method and the JAB intervals have comparable performance and are conservative in terms of coverage, whereas the na\"ive interval undercovers.
	
	\begin{center}
		\begin{table}[H]
			\centering
		\begin{tabular}{||p{3cm}|c|c|c|c|c|c||}
			\hline
			& & & & & & \\
			\textbf{Dataset} & 
			$\ErrOOB$
			& 
			True Error 
			& $\ErrOOB$
			& 
			$\widehat{\rm SE}_\text{na\"ive}$ 
			&  
			$\widehat{\rm SE}_\text{delta}$ 
			&  
			$\widehat{\rm SE}_\text{JAB}$ 

			\\
			& Mean & Mean  & SD & Mean & Mean & Mean\\
			\hline
			Communities \newline ($n = 1994, p = 99$) & .020 & .020 & .0019 & .0020 & .0030  & .0030\\
			Forest Fires \newline ($n = 517, p = 10$) & .4301 & .4556 & .4499  & .3108 & .3569  & .3456\\
			Boston Housing \newline ($n = 506, p = 13$) & 18.85 & 18.58  & 4.13 & 5.69 & 6.56  & 5.76\\
			Servo \newline ($n = 167, p = 2$) & 1.09 & 1.11 & 0.35  & 0.37 & 0.40 & 0.41\\
			\hline	
		\end{tabular}
	
		\caption{\em Results on datasets obtained from the UCI Machine Learning Repository.  Each column is computed over 200 different train-test splits, with $25\%$ in the training set and the rest in the test set.  Each time, we use $B = 3000$ trees in the Random Forest. \label{tab:realdata_reg}}
	\end{table}
	\end{center}

	\begin{figure}[H]
		\centering
		\includegraphics[height=3in]{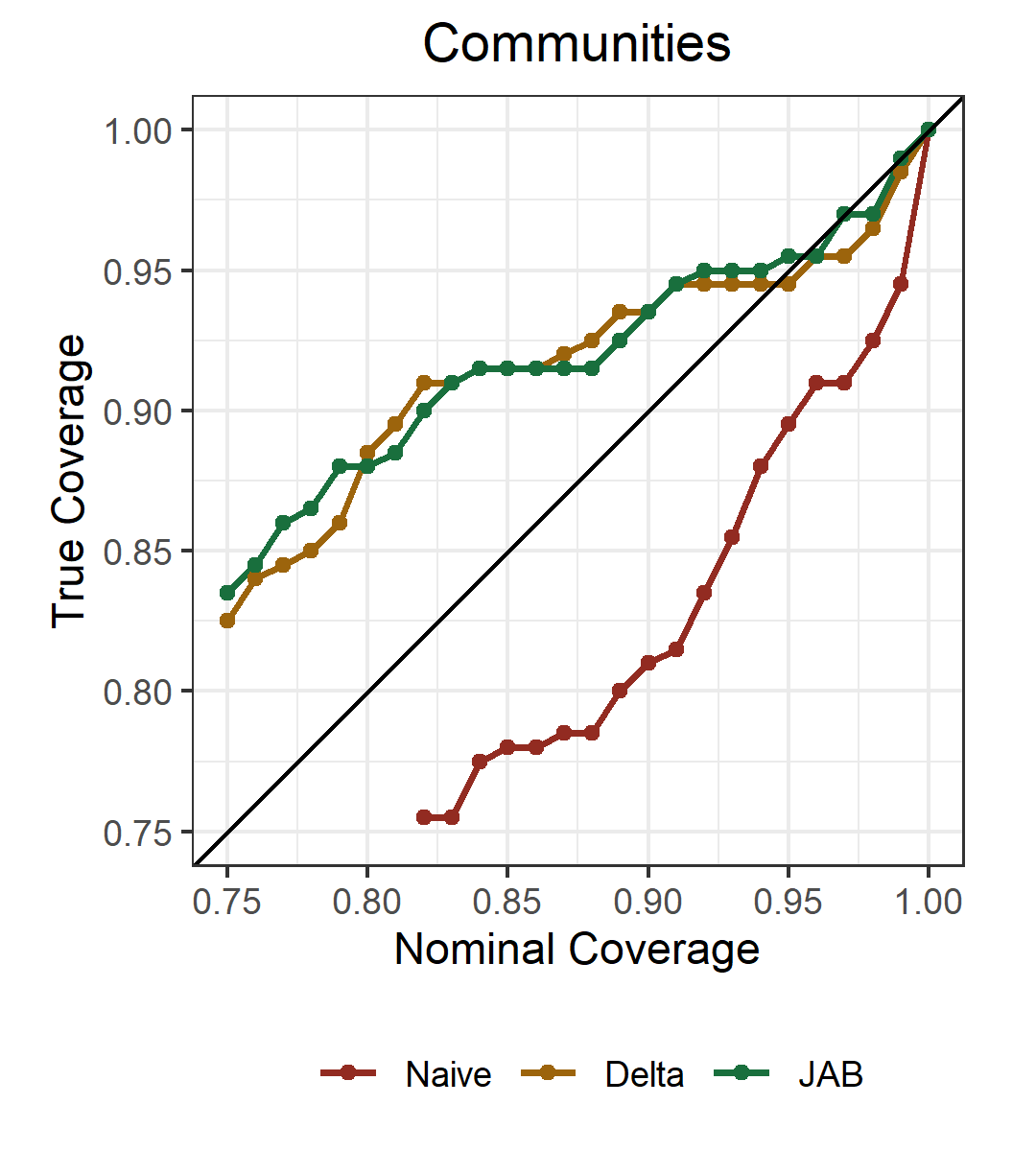}
		\caption{\em Communities dataset, $n = 1994, p = 99$.  Coverage of test error rate over $200$ different train-test splits, with 25\% data in the training set each time.  We fix $B = 3000$ in the Random Forest fit. 
		\label{fig:realdata_reg_coverage}}
	\end{figure}

	\subsection{Simulated Examples: Classification} \label{ss:sim_class}
	
	We now produce the analogous results for classification in Table \ref{tab:results_class}.  In each case, we report the coverage of the  confidence intervals formed using the different standard error estimates, where the nominal coverage is 90\%.  We also report the average true error rate and the average OOB error estimate.
	
\begin{table}[H]
	\begin{center}
		\begin{tabular}{|| *{11}{c}|| }
		\hline
		\multicolumn{3}{||c}{Setting}    
		& \multicolumn{2}{c}{Mean}
		& \multicolumn{3}{c}{Mean CI Width}
		& \multicolumn{3}{c||}{Miscoverage}             \\
		$n$ & $p$ & SNR 
		& $\ErrOOB$ & Truth
		& Na\"ive & Delta & JAB
		& {Na\"ive}
		& {Delta}
		& {JAB}                \\ [0.5ex]
	% latex table generated in R 4.1.1 by xtable 1.8-4 package
	% Sat Jan 15 01:09:54 2022
	\hline
	\hline
	110 & 10 & 0 & 0.27 & 0.27 & .04 & .06 & .10 & \tcr{26.9\%} & \tcr{15.0\%} & 1.0\% \\ 
	\hline
	110 & 100 & 0 & 0.26 & 0.26 & .02 & .04 & .08 & \tcr{28.0\%} & 6.8\% & 0.0\% \\ 
	\hline
	110 & 1000 & 0 & 0.25 & 0.25 & .02 & .03 & .07 & \tcr{25.0\%} & 2.2\% & 0.2\% \\ 
	\hline
	110 & 10 & 2 & 0.19 & 0.19 & .05 & .06 & .09 & \tcr{17.3\%} & \tcr{10.0\%} & 1.2\% \\ 
	\hline
	110 & 100 & 2 & 0.24 & 0.24 & .02 & .04 & .08 & \tcr{25.3\%} & 5.8\% & 0.0\% \\ 
	\hline
	110 & 1000 & 2 & 0.25 & 0.25 & .02 & .03 & .07 & \tcr{27.2\%} & 1.8\% & 0.2\% \\ 
	\hline
	110 & 10 & 10 & 0.16 & 0.16 & .04 & .05 & .08 & \tcr{21.3\%} & \tcr{13.0\%} & 1.5\% \\ 
	\hline
	110 & 100 & 10 & 0.23 & 0.23 & .02 & .03 & .08 & \tcr{25.3\%} & 7.0\% & 0.1\% \\ 
	\hline
	110 & 1000 & 10 & 0.25 & 0.25 & .02 & .03 & .07 & \tcr{25.1\%} & 1.3\% & 0.0\% \\ 
	\hline
	\hline
		\end{tabular}
		
	\end{center}
	
	\caption{\em 
		Coverage under different settings for classification. In each setting, the coverage is computed over $R = 1000$ replicates.  The Random Forest algorithm is fixed to use $B = 3000$ trees.  In this example $n = 110$ and $p$ varies, with the signal-to-noise ratio fixed to be $2.0$.  The true error rate is estimated using a test set of size $n_{\rm test} = 11k$.  CIs are generated with 10\% nominal miscoverage.
	\label{tab:results_class}}
	
\end{table}

	\begin{figure}[H]
		\centering
		\includegraphics[height=2.6in]{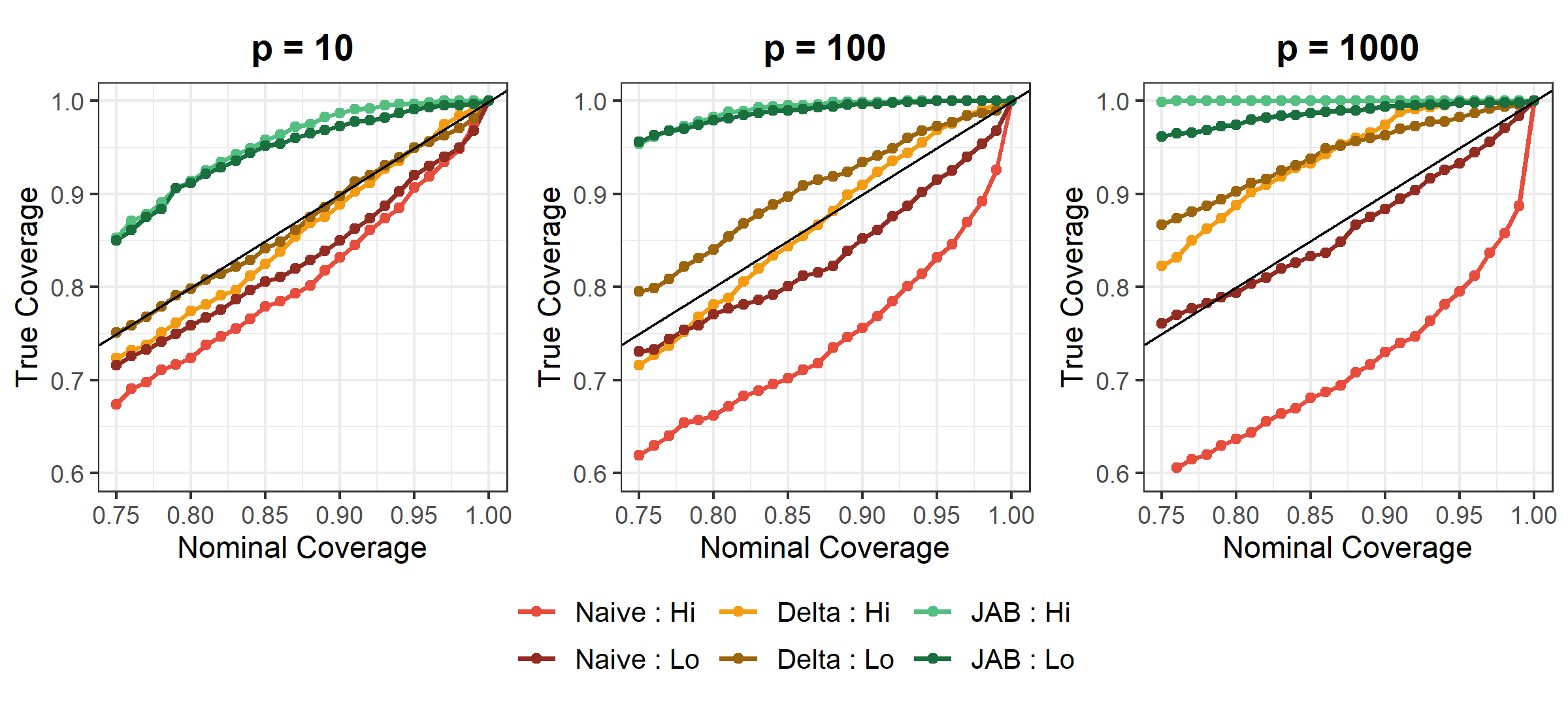}
		\caption{\em Coverage plots under different settings for Classification. In each setting, the coverage is computed over $R = 1000$ replicates.  The Random Forest algorithm is fixed to use $B = 3000$ trees.  In this example $n = 110$ and $p$ varies, with the signal-to-noise ratio fixed to be $2.00$.  The true error rate is estimated using a test set of size $n_{\rm test} = 11k$.  
			\label{fig:class_cov}	}
		
	\end{figure}

	\begin{figure}[H]
		\centering
		\includegraphics[height=2.6in]{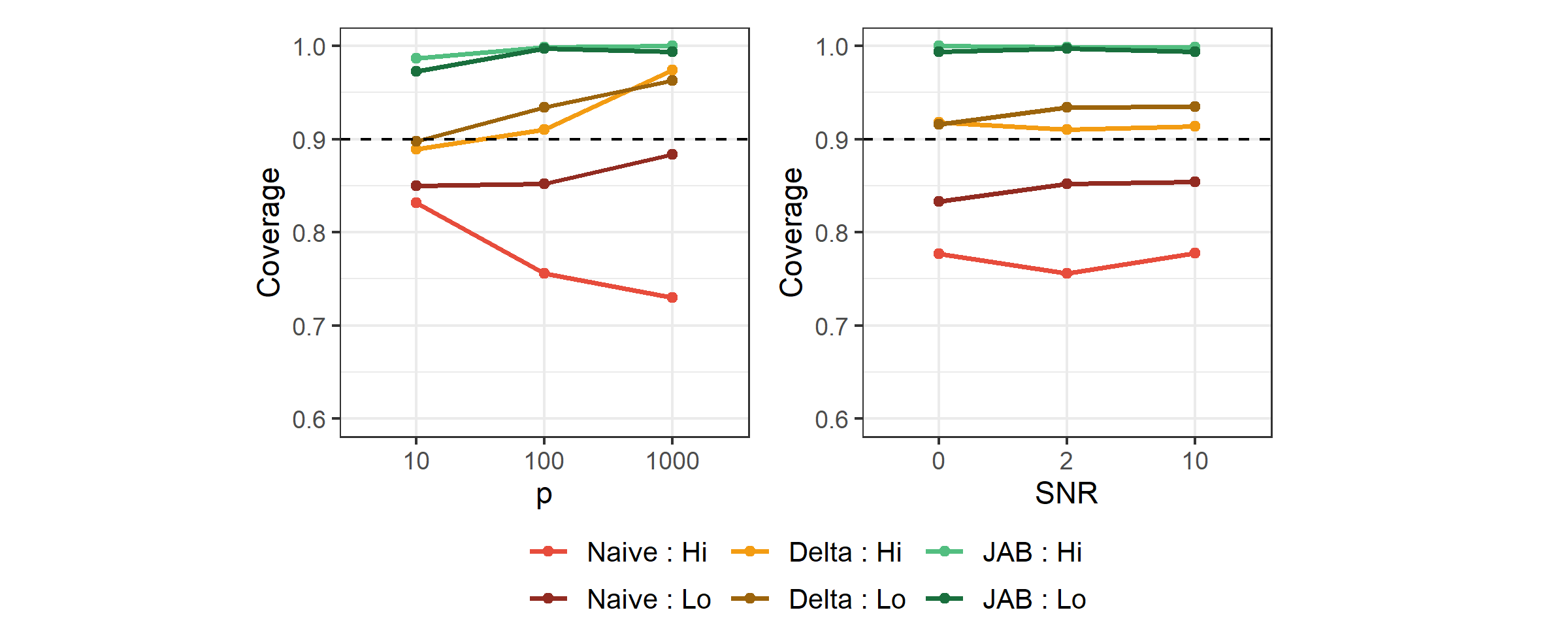}
		\caption{\em Trends across varying parameters $p$, $B$, and SNR for Classification. In each setting, $n = 110$ and the coverage is computed over $R = 1000$ replicates.  
			The true error rate is estimated using a test set of size $n_{\rm test} = 11k$.  Left to right: (1) $B = 3000$, SNR $ = 2.0$, and varying $p$; and (2) $p = 100$, $B = 3000$, and varying SNR
			\label{fig:class_lines}	}	
	\end{figure}

	\section{Conclusion} \label{sec:conclusion}

	We have proposed two new methods for constructing confidence intervals for  the test error in Random Forests, and  demonstrated their utility on real and simulated datasets.  These new intervals have better coverage properties than the na\"ive interval and can be computed with no additional resampling or growing of trees.
	
	We provide R code for implementing the new confidence intervals in Appendix \ref{app:code}, as well as on GitHub at \texttt{https://github.com/RSamyak/oobdelta\_results/}.  We provide reference code for the R packages \texttt{randomForest} and \texttt{ranger}, but the code works with any package that produces individual predictions for each observation from each tree in the Random Forest.

	\section*{Acknowledgements}
	We would like to thank Sourav Chatterjee and Stefan Wager for very helpful conversations.  
	Trevor Hastie was partially supported by grants DMS-2013736 And IIS
	1837931 from the National Science Foundation, and grant 5R01 EB
	001988-21 from the National Institutes of Health.  
	Robert Tibshirani was supported by grant 5R01 EB 001988-16 from the National Institutes of Health and grant DMS1208164 from the National Science Foundation.
	
	\bibliography{fork-ref}
	
	\appendix
	
	\section{Proofs} \label{app:proofs}
		
	\begin{proof}[Proof of Theorem \ref{lem:Di}]
	
	We first observe
	\begin{equation}
	S\l(\F_\epsi\r) =
	 \sum_{j=1}^{n}{
	 	f_\epsi(j)
	 	\l(y_j-\frac{\sum_{b}{
	 			{\widehat{y_j}}^{(b)}I_j^{(b)}g_{\varepsilon,i}(b)}}
 			{\sum_{b}\ I_j^{\left(b\right)}g_{\varepsilon,i}(b)}\right)^2}
	\end{equation}
	where $f_\epsi$ and $g_\epsi$ are the probability mass functions of the observations and the bootstrap samples respectively:
	\begin{align}
		f_{\varepsilon,i}\left(j\right) &=
		\left[1+\left(n\delta_{ij}-1\right)\varepsilon\right]\cdot f_0(j) \quad \text{where}~
		f_0\left(j\right)=\frac{1}{n}\\
		g_{\varepsilon,i}\left(b\right)&=
		\left[\left(1-\varepsilon\right)^n\left(1+\frac{n\varepsilon}{1-\varepsilon}\right)^{N_i^{(b)}}\right]\cdot g_0(b)
		\quad \text{where}~
		g_0\left(b\right)=\frac{1}{n^n}
	\end{align}
	Here $\delta_{ij}$ is the Kronecker delta $\One{(i=j)}$.
	
	We will need the partial derivatives of $f_\epsi$ and $g_\epsi$ at $\varepsilon = 0$:
	\begin{align}
		\dfrac{\partial f_{\varepsilon,i}\left(j\right)}{\partial \varepsilon}|_{\varepsilon = 0} &= 
		\left(\delta_{ij}-\frac{1}{n}\right) \\
		\dfrac{\partial g_{\varepsilon,i}\left(b\right)}{\partial \varepsilon}|_{\varepsilon = 0} &=
		n\left(N_i^{\left(b\right)}-1\right)g_0(b)
	\end{align}
	
	We calculate the empirical influence function using the product rule
	\begin{equation}
		{\hat{D}}_i=\frac{1}{n}\frac{\partial S(\F_\epsi)}{\partial\varepsilon}|_{\varepsilon=0}={\rm I+II}
	\end{equation}
	
	where
	\begin{align*}
		{\rm I} &= 
		\dfrac{1}{n}\sum_{j}{\dfrac{\partial f_\epsi(j)}{\partial \varepsilon} 
		{\l(y_j-\frac{\sum_{b}{
				{\widehat{y_j}}^{(b)}I_j^{(b)}g_{\varepsilon,i}(b)}}
		{\sum_{b}\ I_j^{\left(b\right)}g_{\varepsilon,i}(b)}\right)^2} }|_{\varepsilon = 0}
		\\ &=\frac{1}{n}\sum_{j}{\l(\delta_{ij}-\frac{1}{n}\right)\l(y_j-{\widehat{y_j}}^{OOB}\right)^2} \\
		&=\frac{1}{n}\l\{\l(y_i-{\widehat{y_i}}^{OOB}\right)^2-\dfrac{1}{n}\sum_j\l(y_j-{\widehat{y_j}}^{OOB}\right)^2\r\}
	\end{align*}
	
	and 
	\begin{align*}
		{\rm II} &=
		\dfrac{1}{n}\sum_{j}{f_\epsi(j)} \dfrac{\partial}{\partial\varepsilon}
			{\l(y_j-\frac{\sum_{b}{
						{\widehat{y_j}}^{(b)}I_j^{(b)}g_{\varepsilon,i}(b)}}
				{\sum_{b}\ I_j^{\left(b\right)}g_{\varepsilon,i}(b)}\r)^2} |_{\varepsilon = 0} \\
		&=
		-\frac{2}{n^2}\sum_{j}\l(y_j-{\widehat{y_j}}^{\OOB}\right)\l\{\frac{\sum_{b}{{\widehat{y_j}}^{(b)}I_j^{(b)}}n\l(N_i^{\l(b\r)}-1\r)g_0(b)\mathrm{\ }\ }{e_n^{-1}}-{\widehat{y_j}}^{\OOB}\frac{\sum_{b}\ I_j^{(b)}n\l(N_i^{\l(b\r)}-1\r)g_0(b)}{e_n^{-1}}\r\} \\
		&=-\frac{2e_n}{n}\sum_{j}{\l(y_j-{\widehat{y_j}}^{\OOB}\right)\cdot\l\{\frac{1}{B}\sum_{b}{\l(N_i^{\l(b\r)}-1\r)I_j^{\l(b\r)}\l({\widehat{y_j}}^{\l(b\r)}-{\widehat{y_j}}^{\OOB}\right)}\right\}}
	\end{align*}
	
	We combine them to get 
	\begin{multline}
		{\hat{D}}_i= {\rm I+II}
		=\frac{1}{n}\l\{\l(y_i-{\widehat{y_i}}^{\OOB}\right)^2-\dfrac{1}{n}\sum_j\l(y_j-{\widehat{y_j}}^{\OOB}\right)^2\r\} \\
		-\frac{2e_n}{n}\sum_{j}{\l(y_j-{\widehat{y_j}}^{\OOB}\right)\cdot\l\{\frac{1}{B}\sum_{b}{\l(N_i^{\l(b\r)}-1\r)I_j^{\l(b\r)}\l({\widehat{y_j}}^{\l(b\r)}-{\widehat{y_j}}^{\OOB}\right)}\right\}}
	\end{multline}
	
	\end{proof}
	
	\begin{proof}[Proof of Theorem \ref{lem:arb_loss}]
	Follows in exactly the same fashion as that of Theorem \ref{lem:Di}, replacing the square loss term $\l(y_i-{\widehat{y_i}}^{\OOB}\right)^2$ by $\ell\l(y_i, {\widehat{y_i}}^{\OOB} \r)$.
	\end{proof}
	
	\newpage
	\section{R code for delta-method-after-bootstrap and jackknife-after-bootstrap SE}

	\label{app:code}
	
	\begin{Verbatim}

reduce_function <- einsum::einsum_generator('ib,jb->ib')

mean.sq.diff <- function(u) {
	r <- length(u)
	
	# u <- u[!is.na(u)]
	u[is.na(u)] <- 0
	
	mu <- mean(u, na.rm = TRUE)
	
	ret <- sum((u - mu) ** 2, na.rm = TRUE)
	ret <- ret * (r - 1) / r
	ret <- sqrt(ret)
	return(ret)
}

oobsd_delta_raw.ranger <- function(fit, x = NULL, y, tree_error = NULL, ...) {
	if (is.null(fit$inbag)) {
		stop("fit does not contain inbag \n
		Please run ranger with keep.inbag = TRUE")
	}
	
	inbag <- matrix(unlist(fit$inbag), ncol = fit$num.trees)
	
	if(is.null(tree_error)){
		if(is.null(x)) stop("need either x or tree_error")
		
		data <- data.frame(x, y)
		
		tree_error <-
		(predict(fit,data=data,predict.all=TRUE)$predictions - as.vector(y))
		tree_error[inbag != 0] <- 0
		
	}
	
	
	
	
	sd_delta_internal(tree_error, inbag, fit$predictions, as.vector(y))
}

oobsd_delta_raw.randomForest <- function(fit, x = NULL, y = NULL, tree_error = NULL, ...) {
	if (is.null(fit$inbag)) {
		stop("fit does not contain inbag \n
		Please run randomForest with keep.inbag = TRUE")
	}
	
	if(is.null(y)){
		y <- as.vector(fit$y)
	}
	
	if(is.null(tree_error)){
		if(is.null(x)) stop("need either x or tree_error")
		
		if (is.null(fit$forest))
		stop(
		"fit object does not contain forest! \n
		Please run randomForest(...) with keep.forest = TRUE \n
		and keep.inbag = TRUE"
		)
		
		tree_error <-
		(predict(fit, x, predict.all = TRUE)$individual - y)
		
		tree_error[fit$inbag != 0] <- 0
		
	}
	
	sd_delta_internal(tree_error, fit$inbag, fit$predicted, y)
}


oobsd_delta.ranger <- function(fit, x, y, ...) {
	pmax(oobsd_delta_raw.ranger(fit, x, y, ...),
	oobsd_naive.ranger(fit, y = y, ...))
}

oobsd_delta.randomForest <- function(fit, x, ...) {
	pmax(oobsd_delta_raw.randomForest(fit, x, ...),
	oobsd_naive.randomForest(fit, x, ...))
}


oobsd_naive.ranger <- function(fit, y, ...) {
	sd((fit$predictions - y) ** 2) / sqrt(length(y))
}

oobsd_naive.randomForest <- function(fit, ...) {
	sd((fit$predicted - fit$y) ** 2) / sqrt(length(fit$y))
}


sd_delta_internal <- function(tree_error, inbag, average_prediction, y){
	
	N <- nrow(inbag)
	B <- ncol(inbag)
	
	in_sample <- !(inbag == 0)
	sum_in_sample <- apply(in_sample, 1, sum)
	sum_out_of_bag <- B - sum_in_sample
	
	
	## This is equal to sapply(1:N, function(i){mean(tree_error[i, inbag[i,]==0])})
	average_error <- y - average_prediction
	
	multiplication_ratio <- sum_out_of_bag / (sum_out_of_bag - 1)
	
	
	mse_oob <- mean(average_error ** 2)
	
	approx_e <- function(N) {
		exp(-N * log(1 - 1 / N))
	}
	eN <- approx_e(N)
	
	
	tree_minus_average <- (-1) * (tree_error - average_error)
	
	N_average <- apply(inbag, 1, mean)
	
	inbag_minus_N_average <- inbag - N_average
	
	
	Di_I <- (average_error ** 2 - mse_oob) / N
	
	cross_term <- average_error * tree_minus_average * in_sample
		
	
	reduced <- reduce_function(inbag_minus_N_average, cross_term)
	Di_II <- (-1) * 2 * eN / B * apply(reduced, 1, sum) / N
		
	Di <- Di_I + Di_II
	
	SEhat.noadj <- sqrt(sum(Di ** 2))
	

	
	return(SEhat.noadj)
}


oobsd_jack.ranger <- function(fit, x = NULL, y, all_preds = NULL, ...) {
	if (is.null(fit$inbag)) {
		stop("fit does not contain inbag \n
		Please run ranger with keep.inbag = TRUE")
	}
	
	inbag <- matrix(unlist(fit$inbag), ncol = fit$num.trees)
	
	if(is.null(all_preds)){
		if(is.null(x)) stop("need either x or all_preds")
		
		data <- data.frame(x, y)
		
		all_preds <-
		predict(fit,data=data,predict.all=TRUE)$predictions
		
	}
	
	sd_jack_internal(all_preds, inbag, fit$predictions, as.vector(y))
}

oobsd_jack.randomForest <- function(fit, x = NULL, y = NULL, all_preds = NULL, ...) {
	if (is.null(fit$inbag)) {
		stop("fit does not contain inbag \n
		Please run randomForest with keep.inbag = TRUE")
	}
	
	if(is.null(y)){
		y <- as.vector(fit$y)
	}
	
	if(is.null(all_preds)){
		if(is.null(x)) stop("need either x or all_preds")
		
		if (is.null(fit$forest))
		stop(
		"fit object does not contain forest! \n
		Please run randomForest(...) with keep.forest = TRUE \n
		and keep.inbag = TRUE"
		)
		
		all_preds <-
		predict(fit, x, predict.all = TRUE)$individual
		
		
	}
	
	sd_jack_internal(all_preds, fit$inbag, fit$predicted, y)
}


sd_jack_internal <- function(predmat, inbag, average_prediction = NULL, y){
	
	n <- nrow(inbag)
	B <- ncol(inbag)
	oob.predmat <- matrix(0,n,n)
	oob.n <- oob.predmat
	oobstats_list <- list()
	for(b in 1:B){
		oob <- inbag[, b] == 0
		## each column i represents the oob predictions for elements j with i out as well. (except on the diagonal)
		oob.predmat[oob,oob] <- oob.predmat[oob,oob] + predmat[oob,b]# last vector gets recycled
		oob.n[oob,oob] <- oob.n[oob,oob] + 1
		
	}
	oob.predmat <- oob.predmat/oob.n
	## The diagonal is the OOB prediction for an observation
	## In column i, the off-diagonal elements are the OBB predictions for those elements, in forests fit without i
	
	errs <- (y-oob.predmat)^2
	oob.error <- mean(diag(errs))
	oobi.error <- (colSums(errs)-diag(errs))/(n-1)
	
	sejack <- sqrt(((n-1)*(n-1)/n)*var(oobi.error))
	
	sejack
}

	\end{Verbatim}

\newpage
\section{Extended Simulations} \label{app:results-ext}

We present results for regression analogous to Table \ref{tab:results_reg} with transformed intervals in Tables \ref{tab:results_reg_log} and \ref{tab:results_reg_sqrt}. 

\begin{table}[H]
	\begin{center}
		\begin{tabular}{|| *{11}{c}|| }
			\hline
			\multicolumn{3}{||c}{Setting}    
			& \multicolumn{2}{c}{Mean}
			& \multicolumn{3}{c}{Mean CI Width}
			& \multicolumn{3}{c||}{Miscoverage}             \\
			$n$ & $p$ & SNR
			& $\ErrOOB$ & Truth
			& Na\"ive & Delta & JAB
			& {Na\"ive}
			& {Delta}
			& {JAB}                \\ [0.5ex]
			% latex table generated in R 4.1.1 by xtable 1.8-4 package
			% Wed Jan 26 13:55:13 2022
			\hline
			\hline
			110 & 10 & 0 & 1.1 & 1.1 & .48 & .55 & .63 & \tcr{16.4\%} & \tcr{10.8\%} & 6.9\% \\ 
			\hline
			110 & 100 & 0 & 1.0 & 1.0 & .45 & .53 & .56 & \tcr{12.6\%} & 8.3\% & 5.7\% \\ 
			\hline
			110 & 1000 & 0 & 1.0 & 1.0 & .45 & .52 & .53 & \tcr{13.5\%} & 8.5\% & 7.3\% \\ 
			\hline
			110 & 10 & 2 & 76 & 75 & 33 & 38 & 40 & \tcr{13.1\%} & 8.8\% & 8.2\% \\ 
			\hline
			110 & 100 & 2 & 1147 & 1142 & 507 & 586 & 597 & \tcr{10.9\%} & 7.1\% & 6.5\% \\ 
			\hline
			110 & 1000 & 2 & 12421 & 12450 & 5486 & 6324 & 6490 & \tcr{14.0\%} & 8.7\% & 7.5\% \\ 
			\hline
			110 & 10 & 10 & 42 & 42 & 19 & 21 & 21 & \tcr{13.2\%} & 9.3\% & 9.5\% \\ 
			\hline
			110 & 100 & 10 & 804 & 803 & 355 & 411 & 411 & \tcr{12.5\%} & 8.3\% & 8.6\% \\ 
			\hline
			110 & 1000 & 10 & 9062 & 9079 & 3998 & 4607 & 4712 & \tcr{12.8\%} & 9.1\% & 8.4\% \\ 
			\hline
			\hline
		\end{tabular}
		
	\end{center}
	
	\caption{\em Simulation results for Regression, with intervals using a log transformation.  In each setting, the coverage is computed over $R = 1000$ replicates.  CIs are generated with 10\% nominal miscoverage. \label{tab:results_reg_log}}
	
\end{table}

\begin{table}[H]
	\begin{center}
		\begin{tabular}{|| *{11}{c}|| }
			\hline
			\multicolumn{3}{||c}{Setting}    
			& \multicolumn{2}{c}{Mean}
			& \multicolumn{3}{c}{Mean CI Width}
			& \multicolumn{3}{c||}{Miscoverage}             \\
			$n$ & $p$ & SNR
			& $\ErrOOB$ & Truth
			& Na\"ive & Delta & JAB
			& {Na\"ive}
			& {Delta}
			& {JAB}                \\ [0.5ex]
			% latex table generated in R 4.1.1 by xtable 1.8-4 package
			% Wed Jan 26 14:07:10 2022
			\hline
			\hline
			110 & 10 & 0 & 1.1 & 1.1 & .47 & .54 & .62 & \tcr{16.2\%} & \tcr{11.0\%} & 7.2\% \\ 
			\hline
			110 & 100 & 0 & 1.0 & 1.0 & .45 & .52 & .55 & \tcr{12.9\%} & 8.8\% & 6.3\% \\ 
			\hline
			110 & 1000 & 0 & 1.0 & 1.0 & .44 & .51 & .53 & \tcr{14.2\%} & 8.3\% & 7.6\% \\ 
			\hline
			110 & 10 & 2 & 76 & 75 & 33 & 38 & 39 & \tcr{13.2\%} & 8.1\% & 8.1\% \\ 
			\hline
			110 & 100 & 2 & 1147 & 1142 & 503 & 580 & 591 & \tcr{11.7\%} & 6.9\% & 6.2\% \\ 
			\hline
			110 & 1000 & 2 & 12421 & 12450 & 5441 & 6255 & 6417 & \tcr{14.2\%} & 8.6\% & 7.6\% \\ 
			\hline
			110 & 10 & 10 & 42 & 42 & 18 & 21 & 20 & \tcr{13.2\%} & 8.8\% & 9.5\% \\ 
			\hline
			110 & 100 & 10 & 804 & 803 & 352 & 406 & 406 & \tcr{13.5\%} & 8.8\% & 8.1\% \\ 
			\hline
			110 & 1000 & 10 & 9062 & 9079 & 3965 & 4556 & 4660 & \tcr{13.1\%} & 9.1\% & 8.0\% \\ 
			\hline
			\hline
		\end{tabular}
		
	\end{center}
	
	\caption{\em Simulation results for Regression, with intervals using a square-root transformation.  In each setting, the coverage is computed over $R = 1000$ replicates.  CIs are generated with 10\% nominal miscoverage. \label{tab:results_reg_sqrt}}
	
\end{table}

\newpage

We present results for classification analogous to Table \ref{tab:results_class} with transformed intervals in Tables \ref{tab:results_class_log} and \ref{tab:results_class_sqrt}. 

\begin{table}[H]
	\begin{center}
		\begin{tabular}{|| *{11}{c}|| }
			\hline
			\multicolumn{3}{||c}{Setting}    
			& \multicolumn{2}{c}{Mean}
			& \multicolumn{3}{c}{Mean CI Width}
			& \multicolumn{3}{c||}{Miscoverage}             \\
			$n$ & $p$ & SNR
			& $\ErrOOB$ & Truth
			& Na\"ive & Delta & JAB
			& {Na\"ive}
			& {Delta}
			& {JAB}                \\ [0.5ex]
			% latex table generated in R 4.1.1 by xtable 1.8-4 package
			% Wed Jan 26 14:04:12 2022
			\hline
			\hline
			110 & 10 & 0 & 0.27 & 0.27 & .04 & .06 & .10 & \tcr{26.0\%} & \tcr{14.8\%} & 0.7\% \\ 
			\hline
			110 & 100 & 0 & 0.26 & 0.26 & .02 & .04 & .08 & \tcr{28.0\%} & 6.7\% & 0.0\% \\ 
			\hline
			110 & 1000 & 0 & 0.25 & 0.25 & .02 & .03 & .07 & \tcr{25.0\%} & 1.8\% & 0.1\% \\ 
			\hline
			110 & 10 & 2 & 0.19 & 0.19 & .05 & .06 & .09 & \tcr{18.3\%} & \tcr{10.7\%} & 0.9\% \\ 
			\hline
			110 & 100 & 2 & 0.24 & 0.24 & .02 & .04 & .08 & \tcr{25.3\%} & 6.3\% & 0.0\% \\ 
			\hline
			110 & 1000 & 2 & 0.25 & 0.25 & .02 & .03 & .07 & \tcr{27.1\%} & 1.7\% & 0.2\% \\ 
			\hline
			110 & 10 & 10 & 0.16 & 0.16 & .04 & .05 & .08 & \tcr{21.2\%} & \tcr{13.2\%} & 1.1\% \\ 
			\hline
			110 & 100 & 10 & 0.23 & 0.23 & .02 & .03 & .08 & \tcr{25.7\%} & 7.3\% & 0.1\% \\ 
			\hline
			110 & 1000 & 10 & 0.25 & 0.25 & .02 & .03 & .07 & \tcr{25.4\%} & 1.2\% & 0.0\% \\ 
			\hline
			\hline
		\end{tabular}
		
	\end{center}
	
	\caption{\em Simulation results for Classification, with intervals using a log transformation.  In each setting, the coverage is computed over $R = 1000$ replicates.  CIs are generated with 10\% nominal miscoverage. \label{tab:results_class_log}}
	
\end{table}

\begin{table}[H]
	\begin{center}
		\begin{tabular}{|| *{11}{c}|| }
			\hline
			\multicolumn{3}{||c}{Setting}    
			& \multicolumn{2}{c}{Mean}
			& \multicolumn{3}{c}{Mean CI Width}
			& \multicolumn{3}{c||}{Miscoverage}             \\
			$n$ & $p$ & SNR
			& $\ErrOOB$ & Truth
			& Na\"ive & Delta & JAB
			& {Na\"ive}
			& {Delta}
			& {JAB}                \\ [0.5ex]
			% latex table generated in R 4.1.1 by xtable 1.8-4 package
			% Wed Jan 26 14:06:03 2022
			\hline
			\hline
			110 & 10 & 0 & 0.27 & 0.27 & .04 & .06 & .10 & \tcr{26.4\%} & \tcr{14.8\%} & 0.7\% \\ 
			\hline
			110 & 100 & 0 & 0.26 & 0.26 & .02 & .04 & .08 & \tcr{27.7\%} & 7.0\% & 0.0\% \\ 
			\hline
			110 & 1000 & 0 & 0.25 & 0.25 & .02 & .03 & .07 & \tcr{25.2\%} & 2.0\% & 0.1\% \\ 
			\hline
			110 & 10 & 2 & 0.19 & 0.19 & .05 & .06 & .09 & \tcr{17.0\%} & \tcr{10.6\%} & 1.0\% \\ 
			\hline
			110 & 100 & 2 & 0.24 & 0.24 & .02 & .04 & .08 & \tcr{25.3\%} & 6.0\% & 0.0\% \\ 
			\hline
			110 & 1000 & 2 & 0.25 & 0.25 & .02 & .03 & .07 & \tcr{27.0\%} & 1.8\% & 0.2\% \\ 
			\hline
			110 & 10 & 10 & 0.16 & 0.16 & .04 & .05 & .08 & \tcr{21.4\%} & \tcr{12.8\%} & 1.3\% \\ 
			\hline
			110 & 100 & 10 & 0.23 & 0.23 & .02 & .03 & .08 & \tcr{25.6\%} & 7.3\% & 0.1\% \\ 
			\hline
			110 & 1000 & 10 & 0.25 & 0.25 & .02 & .03 & .07 & \tcr{25.0\%} & 1.3\% & 0.0\% \\ 
			\hline
			\hline
		\end{tabular}
		
	\end{center}
	
	\caption{\em Simulation results for Classification, with intervals using a square-root transformation.  In each setting, the coverage is computed over $R = 1000$ replicates.  CIs are generated with 10\% nominal miscoverage. \label{tab:results_class_sqrt}}
	
\end{table}

\end{document}